\documentclass[10pt,  twocolumn, twoside, letterpaper, conference, final]{IEEEtran}
\IEEEoverridecommandlockouts
\usepackage{amsmath}

\usepackage{examplep}  %% for \pverb command
\usepackage{graphicx}
\usepackage{tikz}
\usepackage[numbers,square,comma,compress]{natbib}
\usepackage{algorithmicx}
\usepackage[ruled]{algorithm}
\usepackage{algpseudocode}
\usepackage{mathtools}
\usepackage{dsfont}
\usepackage{amsfonts}
\usepackage{amssymb}
\usepackage{balance}
\usepackage[margin=0.75in,bottom=0.95in]{geometry}
\usepackage{booktabs,tabularx}
\usepackage{multirow}
\usepackage{textcomp}
\usepackage[usestackEOL]{stackengine}
\usepackage[hyphens]{url}
\usepackage{pict2e}
\usepackage{color}
\usepackage{enumerate}
\usepackage{float}
\usepackage{mathrsfs}
\usepackage{mathtools}
\usepackage{subcaption}
\usepackage{amsthm}
\usepackage[font=small]{caption}
\usepackage{tikz}
\usetikzlibrary{shapes.arrows,decorations.pathreplacing,decorations.markings,shapes.geometric}
\tikzset{radiation/.style={{decorate,decoration={expanding waves,angle=0,segment length=2pt}}}} %%two above for wifi symbol
\usetikzlibrary{positioning}
\usepackage{graphicx}
\usepackage[normalem]{ulem}
\DeclareMathAlphabet\mathbfcal{OMS}{cmsy}{b}{n}
\newtheorem{lemma}{Lemma}
\newtheorem{corollary}{Corollary}
\newtheorem{fact}{Fact}
\makeatletter
\allowdisplaybreaks
\def\blfootnote{\gdef\@thefnmark{}\@footnotetext}
\makeatother

\begin{document}
\title{Finite Blocklength Secrecy Analysis of Polar and Reed-Muller Codes in BEC Semi-Deterministic Wiretap Channels}
\IEEEoverridecommandlockouts
\author{\IEEEauthorblockN{Mahdi Shakiba-Herfeh, Laura Luzzi, and Arsenia Chorti}
\IEEEauthorblockA{ETIS UMR8051, CY Universite, ENSEA, CNRS, F-95000, Cergy, France\\
Email: \{mahdi.shakiba-herfeh, laura.luzzi, arsenia.chorti\}@ensea.fr}
}
\maketitle
%%%%%%%%%%%%%%%%%%%%%%%%%%%%%%%%%%
%%%%%%%%%%%%%%%%%%%%%%%%%%%%%%%%%%
\begin{abstract} 

%%%%%%%%%%%%%%%%%%%%%%%%%%%%%%%%%%
%%%%%%%%%%%%%%%%%%%%%%%%%%%%%%%%%%
In this paper, we consider a semi-deterministic wiretap channel where the main channel is noiseless and the eavesdropper's channel is a binary erasure channel (BEC). 
We provide a lower bound for the achievable secrecy rates of polar and Reed-Muller codes and compare it to the second order coding rate for the semi-deterministic wiretap channel. To the best of our knowledge, this is the first work which demonstrates the secrecy performance of polar and Reed-Muller codes in short blocklengths. The results show that under a total variation secrecy metric, Reed-Muller codes can achieve secrecy rates very close to the second order approximation rate. On the other hand, we observe a significant gap between the lower bound for the achievable rates of polar codes and the the second order approximation rate for short blocklengths.

\end{abstract}

\section{Introduction} \label{sec:intro}

Various physical layer security (PLS) technologies  \cite{Ersi2016, Mehdi2021} are currently being considered for inclusion in the security protocols of sixth generation systems. The interest in PLS concerns primarily massive scale Internet of things (IoT) networks, which have a very wide range of constraints (computational, power, delay, etc.); for such applications, the revision of the standard paradigm of static security controls is currently being discussed. In the general framework of quality of security (QoSec), in which different slices of future networks could be operating under different security levels \cite{Ersi2021}, it is possible that looser confidentiality constraints could be set for low-end IoT slices, e.g., involving short packet delay constrained verticals. 
In this direction, the development of short blocklength wiretap codes emerges as a topic of potentially high practical impact. 

In the seminal paper by Wyner in 1975 \cite{Wyner75}, it was established that keyless confidential communication between two legitimate parties, referred to as Alice and Bob, is possible in the presence of a passive adversary, referred to as Eve, at rates up to the channel's secrecy capacity. Subsequent works characterized the secrecy capacity of more general wiretap models \cite{Csiszar78,Oggier11}. Of notable interest is the recent contribution \cite{yang2019wiretap} in which non asymptotic information theoretic rates were derived, accounting jointly for reliability and secrecy constraints in finite blocklengths.

%In parallel to the information theoretic studies on characterizing the secrecy capacity, many works have focused on coding for secrecy to realize secure communications with the aid of practical codes. 
In terms of practical implementations, numerous proposals have appeared so far in the literature for long blocklengths.  Secrecy capacity achieving coding schemes have been developed  by employing low density parity check codes (LDPCs) \cite{Rathi13}, polar codes \cite{mahdavifar2011}, and lattice codes \cite{Luzzi14}. We note in passing that some works employ Eve's bit error rate or block error rate as a secrecy metric, albeit without any guarantee of statistical independence between the secret message and Eve’s observation. 

To date, there are only a handful of works on coding for secrecy in short blocklengths. In these studies, the secrecy level is evaluated numerically \cite{nooraiepour2020,Pfister17}; as a result these approaches cannot be extended to longer code lengths, e.g., for more than $32$ bits. Finally, for completeness of presentation we note that in \cite{Harrison18, Harrison19} the properties of the best binary wiretap codes at the finite blocklength were investigated.

As a first step in the direction of bridging the existing literature gap, in this paper we analyse the secrecy rates of polar and Reed-Muller codes in semi-deterministic wiretap binary erasure channels (BEC) at short and medium blocklengths. Following \cite{yang2019wiretap}, we consider the average total variation distance (TVD) as the secrecy metric and we derive two lower bounds on the achievable secrecy rate of these codes. To the best of our knowledge, this work is the first to analyse the secrecy performance of polar and Reed-Muller codes in short blocklengths.

The paper is organized as follows. In Section \ref{preliminaries}, we briefly review the construction of polar and Reed-Muller codes. In Section \ref{System}, we introduce the channel model and the secrecy metrics. In Section \ref{secrecy}, which is the main part of our contribution, we present our methodology to analyse the secrecy of codes. Next, in Section \ref{Numerical}, we present numerical results for the lower bound of the achieved secrecy rates of codes for different blocklengths and compare them with the second order approximation rate. Finally, conclusions and open questions are discussed in Section \ref{Conclusion}.

\section{Notation and Preliminaries}\label{preliminaries}
\subsection{Notation}
We use upper case letters to denote random variables (RVs), e.g., $M$ and $Z$ and lower case letters for specific realizations, e.g., $m$ and $z$. $P_Z$ denotes the probability distribution of $Z$ and $P_{MZ}$ the joint probability distribution of the pair $(M,Z)$. $I(M;Z)$, and $I(M;Z|Y)$ stand for mutual information and conditional mutual information, respectively. We denote by $M^n$ a row vector $(M_1,M_2,...,M_n)$, and $M_i$ denotes the $i^{th}$ element of $M^n$. Moreover, $M_i^j$ denotes the subvector $(M_i,M_{i+1},...,M_j)$. The cardinality of a finite set $\mathcal{Z}$ is denoted by $|\mathcal{Z}|$. Moreover, $M_{i_1,...,i_{j}}$ denotes $[M_{i_1}, M_{i_2}, ..., M_{i_{j}}]$. Finally, for two probability measures $P$ and $Q$ on a set $\mathcal{X}$, the total variation distance (TVD) $d(P,Q)$ is defined as follows
\begin{align}
    d(P,Q)= \frac{1}{2}\sum_{x\in \mathcal{X}}\left\lvert P(x)-Q(x)\right\rvert.
\end{align}

\subsection{Polar and Reed-Muller codes}

Polar codes were introduced by Arikan \cite{arikan2009}, who proved that they achieve the capacity for symmetric binary-input channels with a low encoding and decoding complexity. In \cite{mahdavifar2011}, Mahdavifar and Vardy proved that polar codes can asymptotically achieve the secrecy capacity for a wide range of wiretap channels, provided that the main channel and the eavesdropper's channel are binary-input symmetric and that the eavesdropper's channel is degraded with respect to the main channel. The generator matrix of polar codes with blocklength $n=2^m$ is defined as follows
\begin{align}
G_P = \big(\begin{smallmatrix}
  1 & 0\\
  1 & 1
\end{smallmatrix}\big)^{\otimes m},    
\end{align}
where $\otimes$ denotes Kronecker power.\footnote{In this work, in order to simplify notation we consider the version of polar codes without bit reversal. Note that this doesn't affect performance, although it has an impact on complexity \cite{sarkis2014}.}

On the other hand, Reed-Muller codes were introduced by Muller \cite{ReedMuller54}. In \cite{kudekar2017}, it was shown that Reed-Muller codes achieve the capacity of the BEC. Reed–Muller codes can be defined in different equivalent ways. For completeness of presentation, next we introduce some notation before we define the generator matrix of Reed-Muller codes $G_{RM}$.% Reed-Muller codes are generated by the same generator matrix $G_P$ \cite{arikan2010}. For the sake of simplicity in presentation, we consider $G_{RM}$ which is obtained by sorting the rows of $G_P$ with respect to their Hamming weights, such that the first rows of $G_{RM}$ have the lightest Hamming weights and the last rows have the highest Hamming weights.

For $1\leq i\leq m$, let $v_i$ be a $2^m$-tuple over $GF(2)$, generated by $2^{i}$ repetitions of a vector of length $2^{m-i}$, in which the first half elements are 1 and the rest are 0, e.g.,
\begin{align}
    v_i = [\underbrace{1\, 1\, ... \,1}_{\text{$2^{m-i-1}$}} \underbrace{0\, 0\, ... \,0}_{\text{$2^{m-i-1}$}} \underbrace{1\, 1\, ...\, 1}_{\text{$2^{m-i-1}$}} ... \underbrace{0\, 0\, ... \,0}_{\text{$2^{m-i-1}$}}].
\end{align}
Let $v_0$ be the all-ones $2^m$-tuple over $GF(2)$. We define the following Boolean product of two vectors $a=[a_1\, a_2\, ...\, a_n]$ and $b=[b_1\, b_2\, ... \,b_n]$, $a b = [a_1{\cdot}b_1\, a_2{\cdot}b_2\, ... \,a_n{\cdot}b_n]$, where ``$\cdot$'' denotes the Boolean AND operation.

Using the above, the generator matrix $G_{RM}$ of the Reed-Muller code is defined as follows \cite{lin2001}
%\begin{align}
    %G_{RM} = 
    %\{v_0, v_1, &v_2, ..., v_m, v_1v_2, v_1v_3, \nonumber\\
    %&..., \textit{up to products degree $r$}\},
%\end{align}
\begin{align}
G_{RM}=\begin{pmatrix} v_1v_2...v_{n-1}v_n \\ v_2...v_{n-1}v_n\\\vdots\\ v_1v_2 \\ v_m\\ \vdots\\v_1\\v_0 \end{pmatrix}.
\end{align}
The rows of $G_{RM}$ are sorted by the Hamming weights. The first rows have the lowest Hamming weights and the last rows have the highest Hamming weights.

Given $n$ channel uses of a binary input channel $W: \{0,1\} \to \mathcal{Z}$, the bit-channel $W(Z^n,U_{past}|U_i)$, is defined as the binary input channel that takes a single bit $U_i$ as input and the observation vector $Z^n$ and the past inputs $U_{past}$ as output. Arikan showed that as the blocklength asymptotically increases, the bit-channels $W(Z^n,U^{i-1}|U_i)$ polarize, i.e., are either noiseless or completely noisy. Recently, Abbe and Ye connected Reed-Muller codes to polarization theory and showed that polarization of $W(Z^n,U_{past}|U_i)$ also carries over to Reed-Muller codes \cite{abbe2020}. We call the bit-channels which are almost noiseless the ``good bit-channels"  and the bit-channels which are almost completely noisy the ``poor bit-channels".

For the purpose of channel coding, polar codes select the rows of $G_P$ such that the corresponding bit-channels have the lowest Bhattacharyya parameter (parsing the matrix $G_P$ top-to-down), while the remaining bits are considered as frozen bits and fixed to zero. On the other hand, Reed-Muller codes select the last rows of $G_{RM}$, i.e., the rows with the highest Hamming weights, to transmit information, and fix the remaining bits to zero.%The rate of Reed-Muller codes $RM(m,r)$ is given by $R = \sum_{i=0}^r \frac{\binom{m}{i}}{n}$. For rates $\tilde{R}$ which can not be represented in the aforementioned form, we take a subset of $N\tilde{R}$ rows of $G_n$ with the highest Hamming weights.

\section{System Model}\label{System}
We consider a semi-deterministic wiretap channel where the main channel is noiseless and the wiretap channel is a BEC with erasure probability $p$ (Fig. \ref{Fig:SysModel}). Alice encodes her message $M^k$, drawn from a uniform distribution, into a binary codeword $X^n$ such that Bob observes $Y^n=X^n$ without error, while Eve observes $Z^n$. Moreover, the code rate is defined as $R=\frac{k}{n}$.

\begin{figure}[t]
\centering
\includegraphics[width=0.45\textwidth]{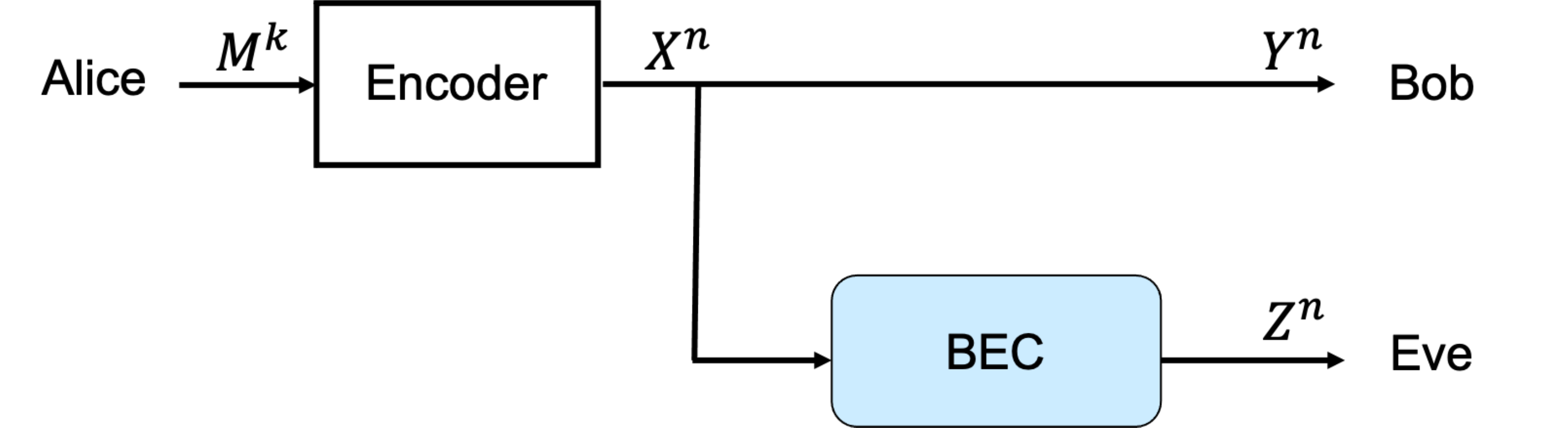}
\caption{The wiretap channel system model.}
\label{Fig:SysModel}
\end{figure}
\vspace{-0.1cm}
\subsection{Coding scheme}
In this paper we consider two wiretap coding schemes based on polar and Reed-Muller codes, following the approach in \cite{mahdavifar2011}. The information bits are sent through the poor bit-channels for Eve. The input vector $U^n$ contains the message vector $M^k$ and the remaining $n-k$ bits are independent and identically distributed (i.i.d.) random bits with uniform distribution. Note that there are no frozen bits in the vector $U^n$ as the direct channel is noiseless. The positions of the message bits  $U^n$ are different in polar and Reed-Muller codes; for polar codes we pick the $k$ rows of $G_P$ such that the corresponding bit-channels have highest Bhattacharyya parameter and for Reed-Muller codes we pick the first $k$ rows of $G_{RM}$, i.e., the rows with the lowest Hamming weights. The relation between the vector $U^n$ and the codeword $X^n$ is as follows
\begin{align}\label{Eq:codeword}
    X^n = U^nG,
\end{align}
where $G$ refers to the generator matrix $G_P$ or $G_{RM}$ respectively, based on the scheme we are employing. We write $i_j$ to denote the position of $M_j$ in $U^n$.

\subsection{Secrecy metric and theoretical limit}
In this paper, similar to \cite{yang2019wiretap}, we choose the average TVD as the secrecy metric under the secrecy condition
\begin{align}
    d(P_{M^kZ^n},\mathcal{Q}^{unif}_MP_{Z^n})<\delta,
\end{align}
where $\mathcal{Q}_M^{unif}$ denotes the uniform distribution on the set $\mathcal{M}$. In \cite{yang2019wiretap}, the authors derive the complete characterization of the second-order secrecy rate for semi-deterministic wiretap channels as follows,
\begin{align}
\label{Eq:Approx}
    R(n,\delta) = C_s - \sqrt{\frac{V_s}{n}}Q^{-1} \left(\frac{\delta}{1-\epsilon}\right)+\mathcal{O} \left(\frac{\log n}{n}\right),
\end{align}
where $\delta$ and $\epsilon$ refer to secrecy and reliability constraints, $C_s$ is the secrecy capacity, and $Q$ denotes the tail probability of the standard Gaussian distribution. Finally, $V_s$ is the conditional variance of the information density of $P_{XZ}$ defined as:
\begin{align}\label{Eq:Vs}
\resizebox{1\hsize}{!}{$V_s = \sum\limits_{x\in \mathcal{X}}P_X(x)\bigg(\sum\limits_{z\in \mathcal{Z}} P_{Z|X}(z|x) \log^2\frac{P_{Z|X}(z|x)}{P_Z(z)}
    -D(P_{Z|X=x}||P_Z)^2\bigg),$}
\end{align}
where $D(\cdot||\cdot)$ is the Kullback–Leibler divergence.

In the case in which Bob's channel is noiseless and Eve's channel is a BEC($p$), we have $C_s=p$. To calculate $V_s$, note that $D(P_{Z|X=x}||P_Z)=1-p$, so that from  (\ref{Eq:Vs}) we get that $V_s = p(1-p)$. In this paper we consider the perfect reliability scenario, i.e., $\epsilon = 0$.

%\subsection{Polar and Reed-Muller codes}
%Consider the matrix $G_i$ defined as follows, $G_i = \big(\begin{smallmatrix}
 % 1 & 0\\
  %1 & 1
%\end{smallmatrix}\big)^{\otimes n}$, where ${\otimes}$ denotes the Kronecker power. A Specifically, for information transmission purpose, in polar code, we pick the set of vectors of matrix $G_n$ which has the lowest Bhattacharyya parameters and on the other hand, for Reed-Muller code we chose rows with highest hamming weights.

\section{Secrecy Analysis of Polar and Reed-Muller Codes}\label{secrecy}
In this section we derive an upper bound for the average TVD of the message $M^k$ and Eve's channel observation $Z^n$ for general wiretap channels.
\subsection{General bounds for the average total variation distance}
\begin{lemma}
(Bound 1) The following inequality holds
\begin{align}\label{Eq:TVD}
d(P_{M^kZ^n},\frac{1}{2^k}P_{Z^n}) \leq \sum_{i=1}^k d(P_{M^iZ^n},\frac{1}{2}P_{M^{i-1}Z^n}).
\end{align}
\end{lemma}
\begin{IEEEproof}
%{ \allowdisplaybreaks
\begin{align}\label{Eq:TVD_proof}
&\sum_{i=1}^k d(P_{M^iZ^n},\frac{1}{2}P_{M^{i-1}Z^n}) \nonumber\\
&=\frac{1}{2} \sum_{i=1}^k \sum_{m^k,z^n}\left\lvert P_{M^iZ^n}(m^i,z^n)-\frac{1}{2}P_{M^{i-1}Z^n}(m^{i-1},z^n)\right\rvert \nonumber\\
&=\frac{1}{2} \sum_{i=1}^k \sum_{m^k,z^n} P_{M_{i+1}^k}(m_{i+1}^k) \nonumber\\
& ~~~~~~~~~~~~~~\left\lvert P_{M^iZ^n}(m^i,z^n)-\frac{1}{2}P_{M^{i-1}Z^n}(m^{i-1},z^n)\right\rvert\nonumber\\
&\overset{\text{a}}=\frac{1}{2}\sum_{m^k,z^n} \bigg(\frac{1}{2^{k-1}} \left\lvert  P_{M_1Z^n}(m_1,z^n)-\frac{1}{2}P_{Z^n}(z^n)\right\rvert+ ...\nonumber\\
&+ \frac{1}{2}\left\lvert P_{M^{m-1}Z^n}(m^{m-1},z^n)-\frac{1}{2}P_{M^{k-2}Z^n}(m^{k-2},z^n)\right\rvert  \nonumber\\
&+  \left\lvert P_{M^kZ^n}(m^k,z^n)-\frac{1}{2}P_{M^{k-1}Z^n}(m^{k-1},z^n)\right\rvert\bigg) \nonumber\\
& \overset{\text{b}}\geq \frac{1}{2}\sum_{m^k,z^n} \bigg|\frac{1}{2^{k-1}}P_{M_1Z^n}(m_1,z^n) -\frac{1}{2^{k}}P_{Z^n}(z^n)\nonumber\\  &+\frac{1}{2^{k-2}}P_{M^2Z^n}(m^2,z^n)- \frac{1}{2^{k-1}}P_{M_1Z^n}(m_1,z^n) + ...\nonumber\\ &-\frac{1}{2}P_{M^{k-1}Z^n}(m^{k-1},z^n) + P_{M^{k}Z^n}(m^{k},z^n)\bigg|\nonumber\\
&=\frac{1}{2}\sum_{m^k,z^n} \left\lvert P_{M^k Z^n}(m^k ,z^n)-\frac{1}{2^k}P_{Z^n}(z^n)\right\rvert \nonumber\\
&= d(P_{M^kZ^n},P_{M^k}P_{Z^n}),
\end{align}
where in (a) we combine all sums in a single sum and use $P_{M_i^k} = \frac{1}{2^{k-i+1}}$, while in (b) we use the triangle inequality.
\end{IEEEproof}

\begin{lemma}
The bit-channel $W(Z^nU_{i_1,...,i_{j-1}}|U_{i_{j}})$ is degraded with respect to the bit-channel $W(Z^nU^{i_{j}-1}|U_{i_{j}})$.\footnote{Without loss of generality, we can assume that $i_1 < i_2 < ... <i_k$.}
\end{lemma}
\begin{IEEEproof}
The proof comes from the fact that the extra outputs of bit-channel $W(Z^nU^{i_{j}-1}|U_{i_{j}})$ can be just dismissed to obtain the bit-channel $W(Z^nU_{i_1,...,i_{j-1}}|U_{i_{j}})$.
\end{IEEEproof}

\begin{lemma}
If the channel $Q$ is degraded with respect to channel $W$ (Fig.\ref{Fig:Degraded}), then $d(P_{XY},P_{X}P_{Y}) \geq d(P_{XZ},P_{X}P_{Z})$, where $Y$ and $Z$ are the outputs of channels $W$ and $Q$, respectively.
\end{lemma}
\begin{figure}[t]
\centering
\begin{tikzpicture}[scale=2,
nodetype2/.style={
	rectangle,
	minimum width=10mm,
	minimum height=7mm,
	align=center,
	draw=black,thick,
},
tip2/.style={-latex, thick, shorten >=0.4mm}
]
\matrix[row sep=0.5cm, column sep=1.3cm, ampersand replacement=\&]{
\node (A) {}; \&
\node (W) [nodetype2]  {$W$}; \&
\node (Q) [nodetype2] {$Q'$}; \&
\node (Z)  {}; 
\\};
\draw[->] (A) edge[tip2] node [above] {\footnotesize $X$} (W);
\draw[->] (W) edge[tip2] node [above] {\footnotesize $Y$}  (Q) ;
\draw[->] (Q) edge[tip2] node [above] {\footnotesize $Z$} (Z) ;
\draw[dashed, thick] (-1,-0.3) rectangle (1,0.3);
\draw node at (0,-0.45) {$Q$};
%\draw[dashed] ($(W.north west)+(-0.1,0.1)$) rectangle ($(Q.south east)+(0.1,-0.1)$);
\end{tikzpicture}
\caption{Channel $Q$ is degraded with respect to channel $W$.}
\label{Fig:Degraded}
\end{figure}
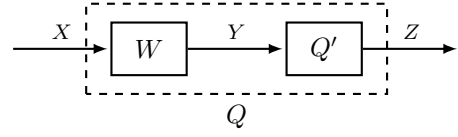
\vspace{-0.05cm}
\begin{IEEEproof}
\begin{align}
    &d(P_{XZ},P_{X}P_Z) = \frac{1}{2} \sum_{x,z}\left\lvert P_{XZ}(x,z)-P_X(x)P_Z(z)\right\rvert\nonumber\\
    &=  \frac{1}{2}\sum_{x} P_X(x) \sum_{z}\left\lvert P_{Z|X}(z|x)-P_Z(z)\right\rvert \nonumber\\
    %&= \frac{1}{2}\sum_{x} P_X(x) \sum_{z}\bigg|\sum_{y}P_{Y|X}(y|x)P_{Z|Y,X}(z|y,x)\nonumber\\
    %&-\sum_{y}P_Y(y)P_{Z|Y}(z|y)\bigg|\nonumber\\
    %&\overset{\text{a}}= \frac{1}{2}\sum_{x} P_X(x) \sum_{z}\bigg|\sum_{y}P_{Y|X}(y|x)P_{Z|Y}(z|y)\nonumber\\
    %&-\sum_{y}P_Y(y)P_{Z|Y}(z|y)\bigg|\nonumber\\
    &=  \frac{1}{2}\sum_{x} P_X(x) \sum_{z}\left\lvert\sum_{y}P_{Z|Y}(z|y)\left(P_{Y|X}(y|x)-P_Y(y)\right)\right\rvert\nonumber\\
    &\overset{\text{a}}\leq  \frac{1}{2}\sum_{x} P_X(x) \sum_{y}\sum_{z}P_{Z|Y}(z|y)\left\lvert P_{Y|X}(y|x)-P_Y(y)\right\rvert\nonumber\\
    &=  \frac{1}{2}\sum_{x} P_X(x) \sum_{y}\left\lvert P_{Y|X}(y|x)-P_Y(y)\right\rvert\nonumber\\
    &=  d(P_{XY},P_XP_Y),
\end{align}
where in (a) we use the triangle inequality.
\end{IEEEproof}

\begin{corollary}
The following inequality holds
\begin{align}
  d(P_{U_{i_1,...,i_{j}}Z^n}&,P_{U_{i_j}}P_{U_{i_1, ...,i_{j-1}}Z^n}) \nonumber\\
  &\leq d(P_{U^{i_j}Z^n},P_{U_{i_j}}P_{U^{i_{j}-1}Z^n})  
\end{align}
\end{corollary}
\begin{proof}
The inequality can be easily shown using Lemma 2 and Lemma 3.
\end{proof}

\begin{corollary}
(Bound 2) The average TVD of $M^k$ and $Z^n$ is upper bounded by the sum of the average TVD of the input-output of bit-channel over the message bits $M_i$.
\begin{align}
  d(P_{M^kZ^n},\frac{1}{2^k}P_{Z^n}) \leq \sum_{j=1}^k d(P_{U^{i_j}Z^n},\frac{1}{2}P_{U^{i_j-1}Z^n}) 
\end{align}
\end{corollary}
\begin{proof}
The proof comes from Lemma 1 and Corollary 1.
\end{proof}
\subsection{Bounds on the TVD and leakage for the BEC}
In this section we specialize the previous bounds to the case of the BEC. First, we note the following property, which holds for any linear code \cite{fazeli2014}.
\begin{fact}
If the channel $W$ is a BEC, the bit channels $W(Z^n,U^{i-1}|U_i)$ are also BECs.
\end{fact}

\begin{lemma}
The average TVD of a BEC($\tilde{p}$) with uniform input $X$ and output $Y$, is $d(P_{XY},\frac{1}{2}P_{Y}) = \frac{1}{2}(1-\tilde{p}) = \frac{1}{2}I(X;Y)$.
\end{lemma}
The proof is straightforward and is omitted. The following Lemma was proven in \cite[Lemma 15]{mahdavifar2011}:
\begin{lemma}
$I(M^k;Z^n)= \sum_{j=1}^k I(U_{i_j};Z^nU_{i_1,..., i_{j-1}}) \leq \sum_{j=1}^kI(U_{i_j};Z^nU^{i_{j}-1})$.
\end{lemma}
From the above Lemma and using Lemma 4, we observe that in the case of the BEC, the leakage is
\begin{align*}
  I(M^k;Z^n) &=2d(P_{Z^nU_{i_1,..., i_{j}}},\frac{1}{2} P_{Z^nU_{i_1,..., i_{j-1}}}) \nonumber\\
  &\leq 2\sum_{j=1}^k d(P_{Z^nU^{i_j}},\frac{1}{2} P_{Z^nU^{i_j-1}}). 
\end{align*}
To the best of our knowledge, there are no results available in the literature for the information-theoretic limits of secrecy rate in finite blocklengths considering the leakage as the secrecy metric. Therefore, in this work the average TVD is used as the secrecy metric.

%From Lemma 5, we observe that a similar equation in Lemma 1, but with an equality sign, holds for $I(M^k;Z^n)$. In addition, the same inequality in Corollary 1 holds for $I(M^k;Z^n)$. Therefore, regarding to the Lemma 4, the presented numerical results on lower bounds on achievable secrecy rates, in Section \ref{Numerical}, with secrecy metric TVD and secrecy condition $\delta$ using Lemma 1 (and Corollary 1), turn to achieved rate (and lower bound on achieved rate) with leakage constraint $2\delta$, respectively. However, bests of our knowledge, there is no study on information theory limit on secrecy rate for finite blocklength considering leakage as the secrecy metric.

\section{Numerical Results}\label{Numerical}
In this section we analyze the secrecy performance of polar and and Reed-Muller codes in BEC semi-deterministic wiretap channels in short and medium blocklengths using numerical simulations.
\subsection{Average TVD of bit-channels for Reed-Muller and polar codes}\label{Numerical_I}
In this subsection, we present the average TVD of each bit-channel $W(Z^nU^{i-1}|U_i)$ for BEC(0.4) and blocklength $n=128$ bits. The erasure probability of bit-channels for polar codes can be calculated by using the closed form expression in \cite{arikan2009}. On the other hand, to derive the erasure probability of bit-channels for Reed-Muller codes, we use Monte-Carlo simulations with $2\cdot 10^5$ channel realizations.
%For Monte-Carlo simulations, we consider $2\cdot 10^5$ channel realizations. Finally, we pick the largest possible set of bit-channels such that the summation of their average TVD satisfies the secrecy constraint.

In our Monte-Carlo simulations, for each channel realization, we consider the vector $Z'$ obtained by omitting the erased bits of the observed vector $Z$ in that channel realization, and the matrix $G_1$ obtained by removing the columns of $G$ that correspond to erased bits. For decoding the bit-channel $i$, as the past bits are assumed to be given, they do not have effect on the uncertainty of the bit-channels. Let $U_{past}$ and $\bar{U}_{past}$ denote the subvectors of $U^n$ corresponding to known and unknown bits, and $G_1^{past}$ and $\bar{G_1}^{past}$ denote the corresponding submatrices of $G_1$. Then
\begin{align}
    \overline{U}_{past}\bar{G_1}^{past} = Z'+U_{past}G_1^{past},
\end{align}
where the right-hand side is known. The bit-channel $i$ is not erased if and only if $U_i$ can be solved from the above equation (Section II.A in \cite{fazeli2014}).

%To do that, for polar code, once we use the close form expression in \cite{arikan2009} and we derive the average TVD of bit-channels $W(Z^nU^{i_{j}-1}|U_{i_{j}})$ and by using Lemma 5 we report a lower bound for achievable secrecy rate. Also, for polar code we employ another method for secrecy analysis and derive the average TVD of bit-channels $W(Z^nU_{i_1,...,i_{j-1}}|U_{i_{j}})$ and by using Lemma 1 we report a tighter lower bound for achievable secrecy rate. The advantage of the former method is that the average TVD of bit-channels $W(Z^nU^{i_{j}-1}|U_{i_{j}})$ for a wide range of channels is efficiently available. However the later method is more accurate but not possible to extend to other channel models. For Reed-Muller code, we find the erasure probability of bit-channels $W(Z^nU_{i_1,...,i_{j-1}}|U_{i_{j}})$ by Monte-Carlo simulation. For Monte-Carlo simulations, we consider $2\cdot 10^5$ channel realizations. Finally, we pick the largest possible set of bit-channels such that the summation of their average TVD satisfies the secrecy constraint.

In Fig. \ref{Fig:TVD} the average TVD of bit-channels for Reed-Muller and polar codes are presented. We observe that for Reed-Muller codes, the good bit-channels are almost concentrated in the last bits and the poor bit-channels are also largely concentrated in the first bits. For polar codes, the concentration is not as dense as for Reed-Muller codes. In Fig. \ref{Fig:TVD_sorted}, the bit numbers are sorted with respect to the average TVD of the bit-channels. In this figure, we observe that Reed-Muller codes have higher polarization speed compared to polar codes, as was shown in \cite{abbe2020}.

\begin{figure}[t]
\centering
\includegraphics[width=0.48\textwidth]{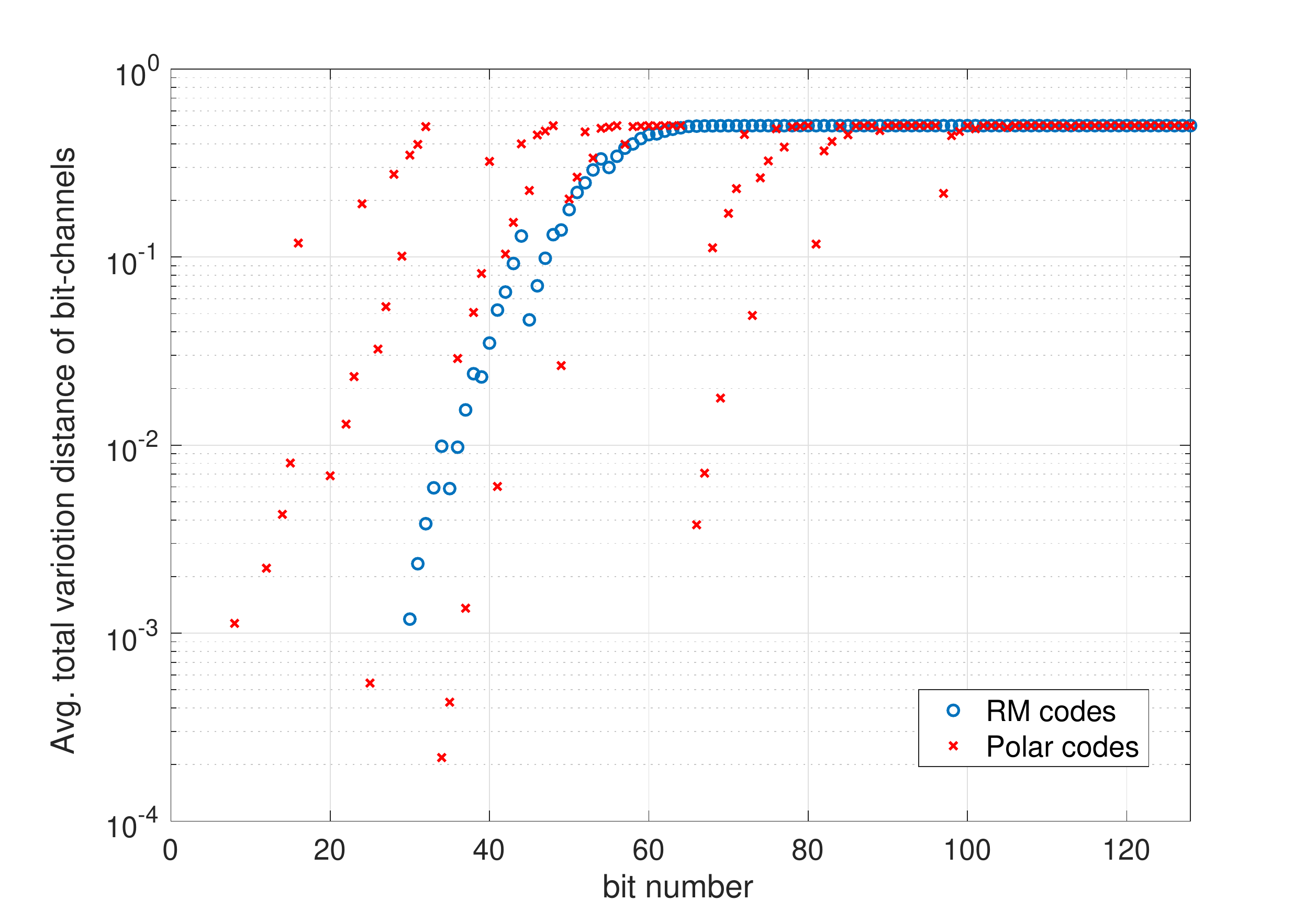}
\caption{The average TVD of bit-channels for BEC with erasure probability $p=0.4$ and $n=128$.}
\label{Fig:TVD}
\end{figure}
\vspace{-0.05cm}
\begin{figure}[t]
\centering
\includegraphics[width=0.48\textwidth]{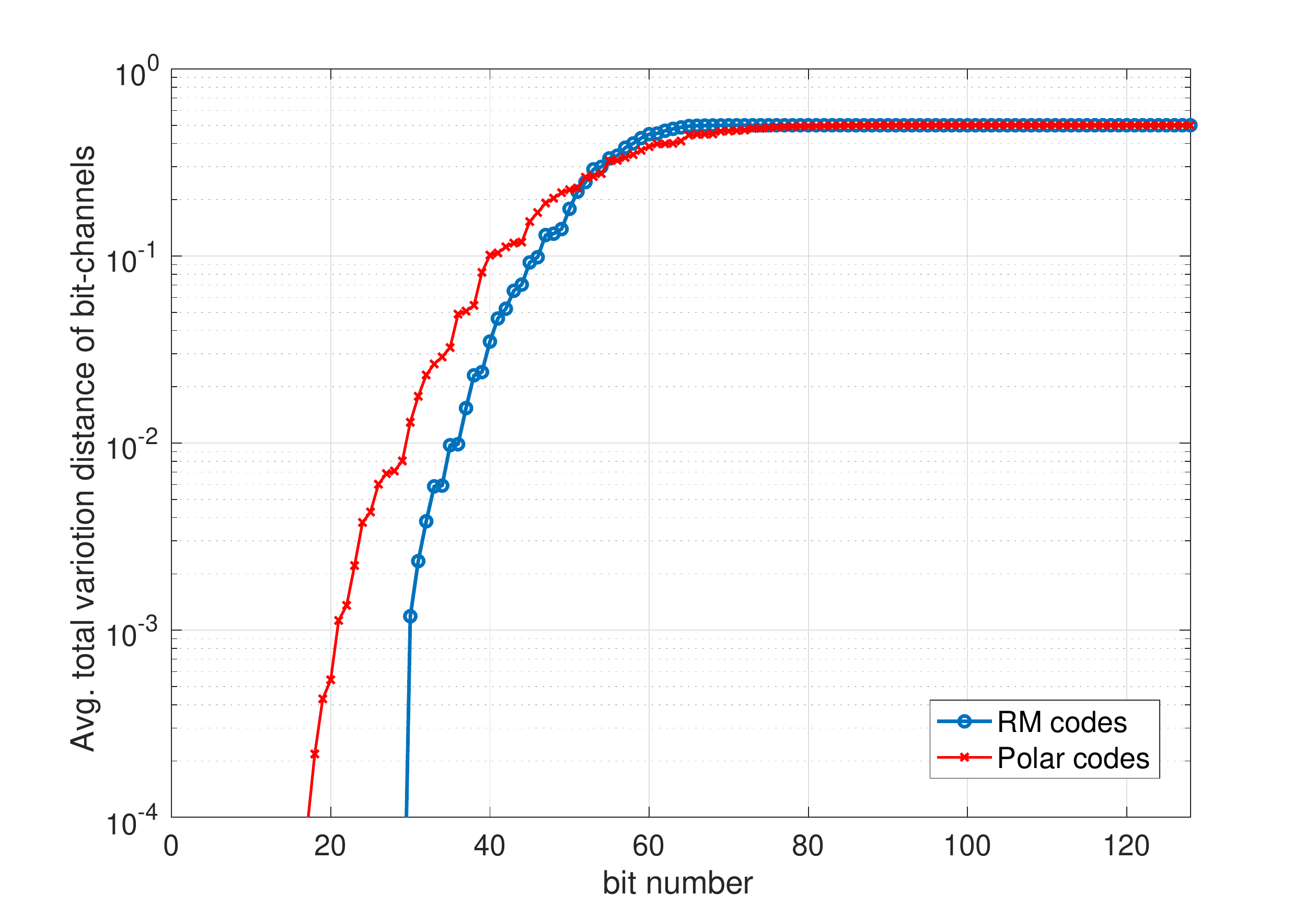}
\caption{The average TVD of bit-channels for BEC with erasure probability $p=0.4$ and $n=128$. The bit numbers are sorted by the average TVD of the bit-channels.}
\label{Fig:TVD_sorted}
\end{figure}
\vspace{-0.1cm}
\subsection{Lower bounds on the achievable secrecy rates}\label{Numerical_II}
%\begin{figure*} [!h]
  %\includegraphics[width=\textwidth,height=8cm]{example04_2.eps}
  %\caption{The comparison of the lower bound on achievable secrecy rates of Reed-Muller and polar codes for BEC semi-deterministic wiretap channel with erasure probability $p=0.4$, with second order approximation secrecy rate (Eq. \ref{Eq:Approx}).}\label{Fig:1}
%\end{figure*}

\begin{figure}[t]
\centering
\includegraphics[width=0.48\textwidth]{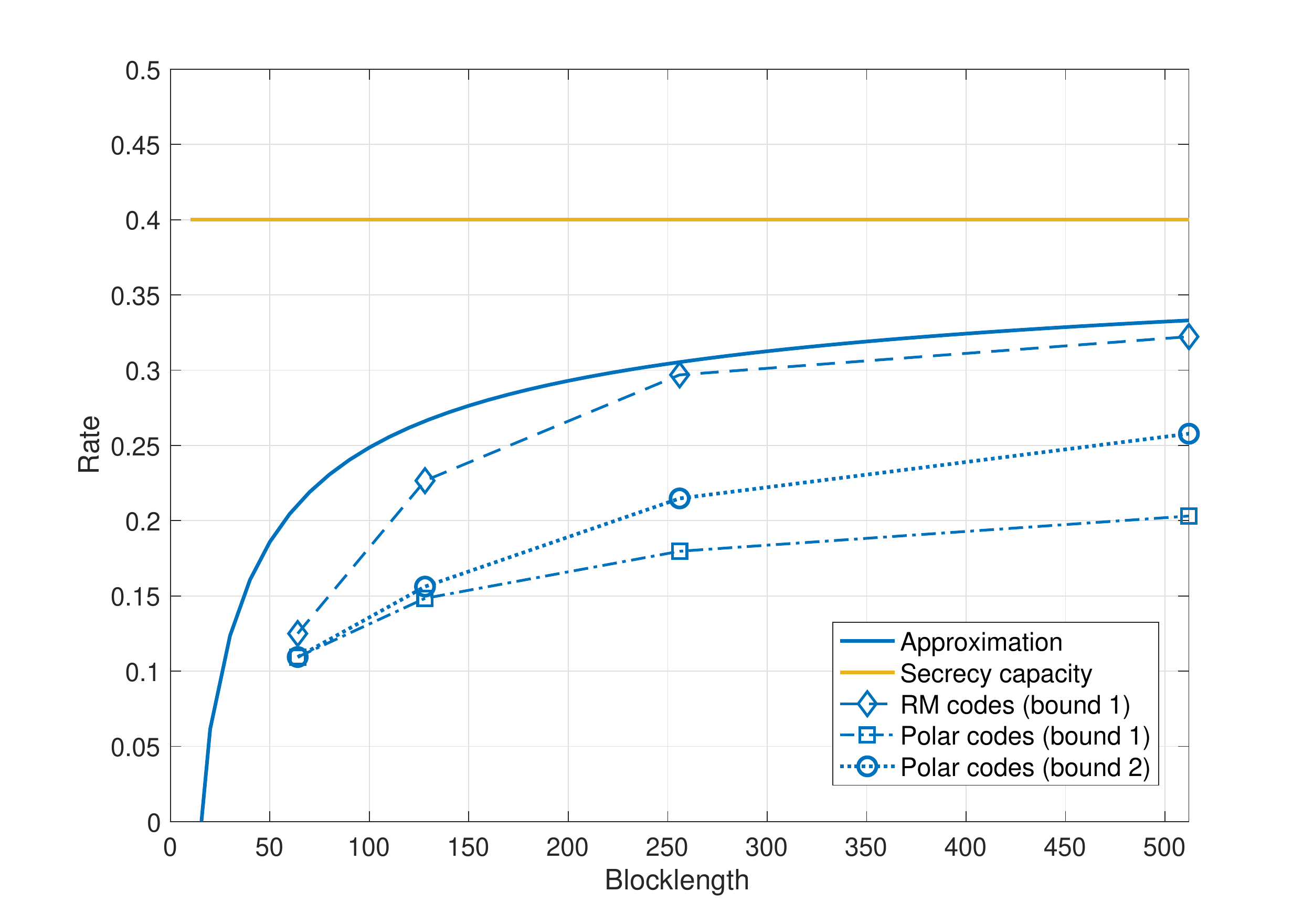}
\caption{Comparison of the lower bound on achievable secrecy rates of Reed-Muller and polar codes for $p=0.4$ and $\delta=0.001$, with the second order approximation secrecy rate in (\ref{Eq:Approx}).}
\label{Fig:1A}
\end{figure}
\vspace{-0.05cm}
\begin{figure}[t]
\centering
\includegraphics[width=0.48\textwidth]{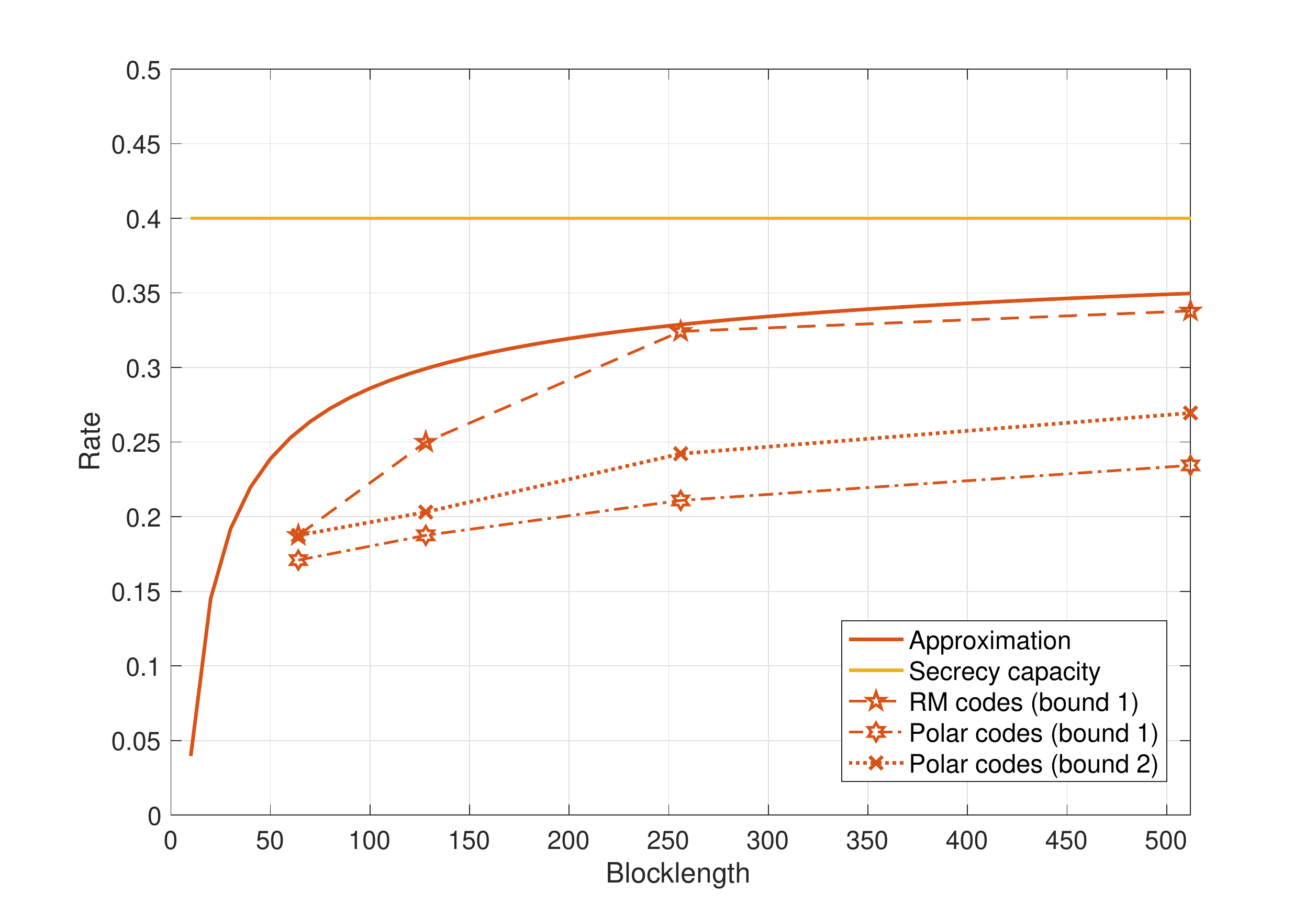}
\caption{Comparison of the lower bound on achievable secrecy rates of Reed-Muller and polar codes for $p=0.4$ and $\delta=0.01$, with the second order approximation secrecy rate in (\ref{Eq:Approx}).}
\label{Fig:1B}
\end{figure}
\vspace{-0.1 cm}
As an example, we consider a BEC(0.4). Figs. \ref{Fig:1A} and \ref{Fig:1B} present the lower bound on the achieved secrecy rates and the second order approximation of achievable secrecy rate for the cases $\delta =0.001$ and $\delta = 0.01$, respectively. We plot the lower bound on the secrecy rate for polar codes using bound 1 and bound 2. We calculate bound 1 by Monte-Carlo simulations and bound 2 by the closed form expression for the bit-channel capacity.

Although bound 1 is more accurate than bound 2 for polar codes, there are some advantages in using bound 2. In fact, using bound 2 we can find the lower bound on the secrecy rate for any value of $\delta$, while for very small values of $\delta$ ($\sim 10^{-9}$), Monte Carlo simulations are not feasible. Moreover, in the case of polar codes, some methods providing good approximations of the bit-channels  $W(Z^nU^{i-1}|U_i)$ for other channel models, are available in the literature \cite{tal2013}. Therefore it is possible to evaluate bound 2 for more general channel models. For Reed-Muller codes the two bounds give the same result, as all the messages are the first top rows of $G_{RM}$, so the two sides in Corollary 1 are equal. We plot the lower bound on the achievable secrecy rates of polar codes using bound 1 by Monte-Carlo simulations.

We observe that Reed-Muller codes show a promising performance in BEC semi-deterministic wiretap channels, specifically for blocklengths 256 and 512. For example at blocklength 256 and $\delta = 0.01$, the second order secrecy rate is approximately 0.329 bits per channel use and the lower bound on achievable secrecy rate of Reed-Muller codes is 0.3242 bits per channel use, i.e., less than 1.5\% away from the second order approximation secrecy rate. On the other hand, there is a significant gap between the achievable rates of polar codes and the second order approximation secrecy rate. This observation verifies the results in \cite{abbe2020}, that Reed-Muller codes polarize faster than polar codes.

%\begin{figure*} [!ht]
  %\includegraphics[width=\textwidth,height=8cm]{example075_2.eps}
  %\caption{The comparison of the achievable secrecy rates of Reed-Muller and polar codes for BEC semi-deterministic wiretap channel with erasure probability $p=0.75$, with second order approximation secrecy rate (Eq. \ref{Eq:Approx}).}\label{Fig:2}
%\end{figure*}

\section{Conclusions and Discussion}\label{Conclusion}
%In this paper, we evaluate the short blocklength secrecy performance of polar and Reed-Muller codes for the BEC semi-deterministic wiretap channel. We consider average TVD as the secrecy metric, and we show that the average TVD between message and Eve's observation is upper bounded by the summation of the average TVD of the bit-channels. We observe that the achievable secrecy rate of Reed-Muller codes are very close to the second order approximation of secrecy rate for blocklengths 256 and 512 bits. 
In this paper, we evaluated the short blocklength secrecy performance of polar and Reed-Muller codes for the BEC semi-deterministic wiretap channels, while the results can be extended to the case where the main channel is noisy. The interest in investigating the secrecy-reliability trade-off of codes is more general, particularly for the application of wiretap codes in IoT scenarios. The evaluation of the achievable secrecy rate (or good approximations) for other codes and more general wiretap models is a timely open problem. Understanding the properties of generator matrices for secrecy coding is fundamental for the design of good short blocklength wiretap codes.
%Based on the lower bound of the achievable secrecy rate in Corollary 2, the secrecy analysis of polar codes for more general cases of wiretap channels is feasible, and it is worthy to be examined in other works. However, analysing the secrecy performance of Reed-Muller codes for more general models, e.g., binary symmetric channel (BSC) and  additive white Gaussian noise (AWGN), is not trivial.

\section*{Acknowledgement}
The authors would like to thank Min Ye for kindly sharing his code for the Monte-Carlo simulation of bit-channel erasure probabilities. This work has been supported by the ANR PRCI project ELIOT (ANR-18-CE40-0030 / FAPESP 2018/12579-7) and the INEX Paris Universite projects PHEBE and eNiGMA.

\balance
\begin{footnotesize}
\bibliographystyle{IEEEtran}

\bibliography{refs}

\begin{thebibliography}{10}
\providecommand{\url}[1]{#1}
\csname url@samestyle\endcsname
\providecommand{\newblock}{\relax}
\providecommand{\bibinfo}[2]{#2}
\providecommand{\BIBentrySTDinterwordspacing}{\spaceskip=0pt\relax}
\providecommand{\BIBentryALTinterwordstretchfactor}{4}
\providecommand{\BIBentryALTinterwordspacing}{\spaceskip=\fontdimen2\font plus
\BIBentryALTinterwordstretchfactor\fontdimen3\font minus
  \fontdimen4\font\relax}
\providecommand{\BIBforeignlanguage}[2]{{%
\expandafter\ifx\csname l@#1\endcsname\relax
\typeout{** WARNING: IEEEtran.bst: No hyphenation pattern has been}%
\typeout{** loaded for the language `#1'. Using the pattern for}%
\typeout{** the default language instead.}%
\else
\language=\csname l@#1\endcsname
\fi
#2}}
\providecommand{\BIBdecl}{\relax}
\BIBdecl

\bibitem{Ersi2016}
A.~Chorti, C.~Hollanti, J.-C. Belfiore, and H.~V. Poor, ``Physical layer
  security: a paradigm shift in data confidentiality,'' in \emph{Physical and
  data-link security techniques for future communication systems}.\hskip 1em
  plus 0.5em minus 0.4em\relax Springer, 2016, pp. 1--15.

\bibitem{Mehdi2021}
M.~Shakiba-Herfeh, A.~Chorti, and H.~V. Poor, ``Physical layer security:
  Authentication, integrity, and confidentiality,'' in \emph{Physical Layer
  Security}.\hskip 1em plus 0.5em minus 0.4em\relax Springer, 2021, pp.
  129--150.

\bibitem{Ersi2021}
\BIBentryALTinterwordspacing
A.~Chorti, A.~N. Barreto, S.~Kopsell, M.~Zoli, M.~Chafii, P.~Sehier,
  G.~Fettweis, and H.~V. Poor, ``Context-aware security for {6G} wireless: The
  role of physical layer security,'' (under review) in \textit{IEEE
  Communications Standards Magazine - SI on Emerging Security Technologies for
  6G}. [Online]. Available: \url{https://arxiv.org/abs/2101.01536}
\BIBentrySTDinterwordspacing

\bibitem{Wyner75}
A.~D. Wyner, ``The wire-tap channel,'' \emph{Bell System Technical Journal},
  vol.~54, no.~8, pp. 1355--1387, 1975.

\bibitem{Csiszar78}
I.~{Csiszar} and J.~{Korner}, ``Broadcast channels with confidential
  messages,'' \emph{IEEE Transactions on Information Theory}, vol.~24, no.~3,
  pp. 339--348, 1978.

\bibitem{Oggier11}
F.~{Oggier} and B.~{Hassibi}, ``The secrecy capacity of the {MIMO} wiretap
  channel,'' \emph{IEEE Transactions on Information Theory}, vol.~57, no.~8,
  pp. 4961--4972, 2011.

\bibitem{yang2019wiretap}
W.~Yang, R.~F. Schaefer, and H.~V. Poor, ``Wiretap channels: Nonasymptotic
  fundamental limits,'' \emph{IEEE Transactions on Information Theory},
  vol.~65, no.~7, pp. 4069--4093, 2019.

\bibitem{Rathi13}
V.~{Rathi}, M.~{Andersson}, R.~{Thobaben}, J.~{Kliewer}, and M.~{Skoglund},
  ``Performance analysis and design of two edge-type {LDPC} codes for the {BEC}
  wiretap channel,'' \emph{IEEE Transactions on Information Theory}, vol.~59,
  no.~2, pp. 1048--1064, 2013.

\bibitem{mahdavifar2011}
H.~Mahdavifar and A.~Vardy, ``Achieving the secrecy capacity of wiretap
  channels using polar codes,'' \emph{IEEE Transactions on Information Theory},
  vol.~57, no.~10, pp. 6428--6443, 2011.

\bibitem{Luzzi14}
C.~{Ling}, L.~{Luzzi}, J.~{Belfiore}, and D.~{Stehlé}, ``Semantically secure
  lattice codes for the {Gaussian} wiretap channel,'' \emph{IEEE Transactions
  on Information Theory}, vol.~60, no.~10, pp. 6399--6416, 2014.

\bibitem{nooraiepour2020}
A.~Nooraiepour, S.~R. Aghdam, and T.~M. Duman, ``On secure communications over
  {Gaussian} wiretap channels via finite-length codes,'' \emph{IEEE
  Communications Letters}, vol.~24, no.~9, pp. 1904--1908, 2020.

\bibitem{Pfister17}
J.~{Pfister}, M.~A.~C. {Gomes}, J.~P. {Vilela}, and W.~K. {Harrison},
  ``Quantifying equivocation for finite blocklength wiretap codes,'' in
  \emph{2017 IEEE International Conference on Communications (ICC)}, 2017, pp.
  1--6.

\bibitem{Harrison18}
W.~K. {Harrison} and M.~R. {Bloch}, ``On dual relationships of secrecy codes,''
  in \emph{56th Annual Allerton Conference on Communication, Control, and
  Computing}, 2018, pp. 366--372.

\bibitem{Harrison19}
W.~K.{ Harrison} and M.~R. {Bloch}, ``Attributes of generators for best finite
  blocklength coset wiretap codes over erasure channels,'' in \emph{IEEE
  International Symposium on Information Theory (ISIT)}, 2019, pp. 827--831.

\bibitem{arikan2009}
E.~Arikan, ``Channel polarization: A method for constructing capacity-achieving
  codes for symmetric binary-input memoryless channels,'' \emph{IEEE
  Transactions on information Theory}, vol.~55, no.~7, pp. 3051--3073, 2009.

\bibitem{sarkis2014}
G.~Sarkis, P.~Giard, A.~Vardy, C.~Thibeault, and W.~J. Gross, ``Fast polar
  decoders: Algorithm and implementation,'' \emph{IEEE Journal on Selected
  Areas in Communications}, vol.~32, no.~5, pp. 946--957, 2014.

\bibitem{ReedMuller54}
D.~E. Muller, ``Application of boolean algebra to switching circuit design and
  to error detection,'' \emph{Transactions of the IRE professional group on
  electronic computers}, no.~3, pp. 6--12, 1954.

\bibitem{kudekar2017}
S.~Kudekar, S.~Kumar, M.~Mondelli, H.~D. Pfister, E.~{\c{S}}a{\c{s}}oǧlu, and
  R.~L. Urbanke, ``{Reed-Muller} codes achieve capacity on erasure channels,''
  \emph{IEEE Transactions on information theory}, vol.~63, no.~7, pp.
  4298--4316, 2017.

\bibitem{lin2001}
S.~Lin and D.~J. Costello, \emph{Error control coding}.\hskip 1em plus 0.5em
  minus 0.4em\relax Prentice hall, 2001, vol.~2, no.~4.

\bibitem{abbe2020}
E.~Abbe and M.~Ye, ``{Reed-Muller} codes polarize,'' \emph{IEEE Transactions on
  Information Theory}, vol.~66, no.~12, pp. 7311--7332, 2020.

\bibitem{fazeli2014}
A.~Fazeli and A.~Vardy, ``On the scaling exponent of binary polarization
  kernels,'' in \emph{2014 52nd Annual Allerton Conference on Communication,
  Control, and Computing (Allerton)}.\hskip 1em plus 0.5em minus 0.4em\relax
  IEEE, 2014, pp. 797--804.

\bibitem{tal2013}
I.~Tal and A.~Vardy, ``How to construct polar codes,'' \emph{IEEE Transactions
  on Information Theory}, vol.~59, no.~10, pp. 6562--6582, 2013.

\end{thebibliography}
\end{footnotesize}

\end{document}